\documentclass[review,12pt]{elsarticle}
\usepackage{amsmath}
\usepackage{amsfonts}
\usepackage{amssymb}
\usepackage{amsthm}
\newtheorem{theorem}{Theorem}[section]
\newtheorem{lemma}[theorem]{Lemma}
\newtheorem{proposition}[theorem]{Proposition}
\newtheorem{corollary}[theorem]{Corollary}
\theoremstyle{definition}

\newtheorem{remark}[theorem]{Remark}

\usepackage{hyperref}
\journal{***}

\begin{document}

\begin{frontmatter}

\title{All $\alpha+u\beta$-constacyclic codes of length $np^{s}$ over $\mathbb{F}_{p^{m}}+u\mathbb{F}_{p^{m}}$}


\author{Wei Zhao\corref{cor1}}
\ead{zhaowei4892@sina.com}

\author{Xilin Tang\corref{cor2}}
\ead{xilintang2016@sina.com}

\author{Ze Gu\corref{cor1}}
\ead{guze528@sina.com}

\cortext[cor2]{Corresponding author, Supported in part by NSFC (No. 11571119)}

\address{Department of Mathematics, South China University of Technology, Guangzhou, Guangdong, 510640, P.R. China}

\begin{abstract}
Let $\mathbb{F}_{p^{m}}$ be a finite field with cardinality $p^{m}$ and $R=\mathbb{F}_{p^{m}}+u\mathbb{F}_{p^{m}}$ with $u^{2}=0$. We aim to determine all
$\alpha+u\beta$-constacyclic codes of length $np^{s}$ over $R$, where $\alpha,\beta\in\mathbb{F}_{p^{m}}^{*}$,
$n, s\in\mathbb{N}_{+}$ and $\gcd(n,p)=1$. Let $\alpha_{0}\in\mathbb{F}_{p^{m}}^{*}$ and $\alpha_{0}^{p^{s}}=\alpha$.
The residue ring $R[x]/\langle x^{np^{s}}-\alpha-u\beta\rangle$ is a chain ring with the maximal ideal $\langle x^{n}-\alpha_{0}\rangle$ in the case that $x^{n}-\alpha_{0}$ is irreducible in $\mathbb{F}_{p^{m}}[x]$. If $x^{n}-\alpha_{0}$
is reducible in $\mathbb{F}_{p^{m}}[x]$, we give the explicit expressions of the ideals of $R[x]/\langle x^{np^{s}}-\alpha-u\beta\rangle$.
Besides, the number of codewords and the dual code of every $\alpha+u\beta$-constacyclic code are provided.
\end{abstract}


\end{frontmatter}

\section{Introduction}

The class of constacyclic codes plays a very significant role in the theory of error-correcting codes [23,25].
As a generalization of cyclic codes, constacyclic codes have practical applications as they
can be efficiently encoded with simple shift registers. Many researchers are  thus interested in this class of
codes for both theoretical and practical reasons.

Let $\mathbb{F}_{p^{m}}$ be a finite field with cardinality $p^{m}$, where $p$ is a prime and $m$ is a
positive integer. For any $\lambda\in\mathbb{F}_{p^{m}}^{*}$, a $\lambda$-constacyclic code of length
$n$ over $\mathbb{F}_{p^{m}}$ is an ideal of $\mathbb{F}_{p^{m}}[x]/\langle x^{n}-\lambda\rangle$. In
the literatures, most of the research is concentrated on the case of $\gcd(n,p)=1$. The case that $n$ is divisible by $p$
yields the so-called repeated-root codes, which were first studied in 1967 by Berman [6] and then
by Massey et al. [26], Falkner et al. [20] and Roth and Seroussi [31] in the 1970s and 1980s. Repeated-root
codes were investigated in the most generality by Castagnoli et al [11] and van Lint [32]. It turns out that
such codes are optimal in a few cases, which motivates researchers to further study this class of codes.

After the realization that many important yet seemingly non-linear codes over finite fields are actually
closely related to linear codes over $\mathbb{Z}_{4}$ in particular, and codes over finite rings in general,
the codes over finite rings have attracted a great deal of attention (see [1, 7, 9, 28]). The classification
of codes plays a very important role in studying their structures and encoders. However, it is a very difficult
task in general and only several codes of special lengths over certain finite fields or finite chain rings are classified.
The classification and the detailed structures of all constacyclic codes of length $2^{s}$ over the Galois extension rings of
$\mathbb{F}_{2}+u\mathbb{F}_{2}$ are given in [13]. In [15], Dinh classified all
constacyclic codes of length $p^{s}$ over $\mathbb{F}_{p^{m}}+u\mathbb{F}_{p^{m}}$. Later, Dinh et al.[16]
studied negacyclic codes of length $2p^{s}$ over the ring $\mathbb{F}_{p^{m}}+u\mathbb{F}_{p^{m}}$.
Chen et al. [12] investigated all constacyclic codes of length $2p^{s}$ over the ring
$\mathbb{F}_{p^{m}}+u\mathbb{F}_{p^{m}}$. For any $\alpha\in\mathbb{F}_{p^{m}}^{*}$, the
$\alpha$-constacyclic codes of length $np^{s}$ over $\mathbb{F}_{p^{m}}+u\mathbb{F}_{p^{m}}$
are provided by Cao et al. in [10]. The purpose of this paper is to
determine the algebraic structures of all $\alpha+u\beta$-constacyclic codes of length $np^{s}$ over
$\mathbb{F}_{p^{m}}+u\mathbb{F}_{p^{m}}$.

The remainders of this paper are organized as follows. Preliminary concepts and results are shown in Section 2.
We assume in Section 3 that $x^{n}-\alpha_{0}$ is irreducible in $\mathbb{F}_{p^{m}}[x]$, where $\alpha_{0}\in \mathbb{F}_{p^{m}}$
 and satisfies $\alpha=\alpha_{0}^{p^{s}}$. It is  shown that
the ambient ring $(\mathbb{F}_{p^{m}}+u\mathbb{F}_{p^{m}})[x]/\langle x^{np^{s}}-\alpha-u\beta\rangle$
is a chain ring with a maximal ideal $\langle x^{n}-\alpha_{0}\rangle$, and thus all
$\alpha+u\beta$-constacyclic codes are $\langle (x^{n}-\alpha_{0})^{i}\rangle$,
$0\leq i\leq 2p^{s}$. In Section 4, we consider the remaining case where $x^{n}-\alpha_{0}$
is reducible in $\mathbb{F}_{p^{m}}[x]$.  According to the factorization of $x^{n}-\alpha_{0}$,
the detailed structures of ideals of the ambient ring
$(\mathbb{F}_{p^{m}}+u\mathbb{F}_{p^{m}})[x]/\langle x^{np^{s}}-\alpha-u\beta\rangle$
are provided. Among other results, we also exhibit the number of $\alpha+u\beta$-constacyclic
codes and the dual of every $\alpha+u\beta$-constacyclic code.

\section{Preliminaries}

Let $R$ be a finite commutative ring with identity $1\neq0$. An ideal $I$ of $R$ is called
$\emph{principal}$ if it is generated by one element. If all the ideals of $R$ are
principal, then $R$ is called a $\emph{principal ideal ring}$. $R$ is called a
$\emph{local ring}$ if $R$ has a unique maximal ideal. A ring $R$ is called a \textit{chain ring} if the set
of all ideals of $R$ is linearly ordered under set inclusion. The following
equivalence conditions are well known (cf.[17]).

\begin{proposition}
~Let $R$ be a finite commutative ring. Then the following conditions are equivalent:

($i$) $R$ is a local ring and the unique maximal ideal $M$ of $R$ is principal, i.e.,
$M=\langle r\rangle$ for some $r\in R$;

($ii$) $R$ is a local principal ideal ring;

($iii$) $R$ is a chain ring whose ideals are $\langle r^{i}\rangle$, $0\leq i\leq N(r)$,
where $N(r)$ is the nilpotency index of $r$.

Moreover, if $R$ is a finite chain ring with the unique maximal ideal $\langle r\rangle$
and the nilpotency index of $r$ is $e$, then the cardinality of the ideal
$\langle r^{i}\rangle$ is $|R/\langle r\rangle|^{e-i}$ for $i=0,1,\ldots,e-1$.
\end{proposition}

A \textit{code} $\mathcal{C}$ of length $n$ over $R$ is a nonempty subset of
$R^{n}$. The code $\mathcal{C}$
is said to be $\emph{linear}$ if $\mathcal{C}$ is a $R$-submodule of $R^{n}$. For a unit
$\lambda$ of $R$, the $\lambda$-constacyclic ($\lambda$-twisted) shift $\tau_{\lambda}$ on
$R^{n}$ is the shift
\begin{center}
$\tau_{\lambda}(x_{0},x_{1},...,x_{n-1})=(\lambda x_{n-1},x_{0},...,x_{n-2})$.
\end{center}
A linear code $\mathcal{C}$ is said to be \textit{$\lambda$-constacyclic} if $\tau_{\lambda}(\mathcal{C})=\mathcal{C}$.
Particularly, $\mathcal{C}$ is called a \textit{cyclic code} if $\lambda=1$, and a \textit{negacyclic code}
if $\lambda=-1$. For any $a=(a_{0},a_{1},..., a_{n-1})\in R^{n}$, let
$a(x)=\sum\limits_{i=0}^{n-1}a_{i}x^{i}\in R[x]/\langle x^{n}-\lambda\rangle$. The codewords of $\mathcal{C}$ are then
identified with the polynomials in $R[x]$. In the ring $R[x]/\langle x^{n}-\lambda\rangle$, $xa(x)$ corresponds to a
$\lambda$-constacyclic shift of $a(x)$. From that, the following proposition is straightforward [23, 25].

\begin{proposition}
~A linear code $\mathcal{C}$ of length $n$ over $R$ is a $\lambda$-constacyclic code if and only if
$\mathcal{C}$ is an ideal of the residue class ring $R[x]/\langle x^{n}-\lambda\rangle$.
\end{proposition}

The inner product of ambient space $R^{n}$ is defined as usual, i.e.,
$\textbf{a}\cdot \textbf{b}=\sum\limits_{i=0}^{n-1}a_{i}b_{i}$, where $\textbf{a}=(a_{0},a_{1},..., a_{n-1})$,
$\textbf{b}=(b_{0},b_{1},...,b_{n-1})\in R^{n}$. For a linear code $\mathcal{C}$ over $R$,
its \emph{dual} code $\mathcal{C}^{\bot}$ is the set of $n$-tuples over $R$ that
are orthogonal to all codewords of $\mathcal{C}$, i.e., $\mathcal{C}^{\bot}=\{\textbf{x}|\textbf{x}\cdot \textbf{y}=0,\forall \textbf{y}\in\mathcal{C}\}$.
A code $\mathcal{C}$ is called $\emph{self-orthogonal}$ if $\mathcal{C}\subseteq \mathcal{C}^{\bot}$,
and it is called \textit{self-dual} if $\mathcal{C}=\mathcal{C}^{\bot}$.
We list some known results which will be used in our paper below.

\begin{proposition}(Proposition 2.11 [17])
~Let $p$ be a prime and $R$ be a finite ring of size $p^{\alpha}$.
The number of codewords in any linear code $\mathcal{C}$ of length $n$ over
$R$ is $p^{k}$, for some integer $k\in\{0,1,...,\alpha n\}$. Moreover,
the dual code $\mathcal{C}^{\bot}$ has $p^{l}$ codewords, where
$k+l=\alpha n$.
\end{proposition}

\begin{proposition}(Proposition 2.4 [13])
~The dual of a $\lambda$-constacyclic code of length $n$ over $R$ is a
$\lambda^{-1}$-constacyclic code of length $n$ over $R$, i.e., an ideal
of $R[x]/\langle x^{n}-\lambda^{-1}\rangle$.
\end{proposition}

In the next, without special instructions, $R$ denotes $\mathbb{F}_{p^{m}}+u\mathbb{F}_{p^{m}}$ with $u^{2}=0$. It is known that $R$ is
a chain ring with the unique maximal ideal $u\mathbb{F}_{p^{m}}$. Let
$\mathcal{R}=R[x]/\langle x^{np^{s}}-\alpha+u\beta\rangle$, where $\alpha+u\beta$ is a unit in $R$.
Clearly, $\alpha+u\beta$ is a unit in $\mathbb{F}_{p^{m}}+u\mathbb{F}_{p^{m}}$
if and only if $\alpha\in\mathbb{F}_{p^{m}}^{*}$. It follows from Proposition 2.2
that $\alpha+u\beta$-constacyclic codes of length $np^{s}$ over $R$ are ideals of
$\mathcal{R}$. When $\beta=0$, the structures of $\alpha$-constacyclic codes
of length $np^{s}$ over $R$ are provided by Cao et al. in [10]. In this paper, we
always assume $\alpha,\beta\neq0$. For any $\alpha\in\mathbb{F}_{p^{m}}^{*}$, there exist
$\alpha_{0}\in\mathbb{F}_{p^{m}}^{*}$ such that $\alpha_{0}^{p^{s}}=\alpha$. The
structures of $\alpha+u\beta$-constacyclic codes of length $np^{s}$ over $R$ are
dependent on whether $x^{n}-\alpha_{0}$ is irreducible or not in $\mathbb{F}_{p^{m}}[x]$.
From [24], we have the following result.

\begin{lemma}
Let $n\geq2$ be an integer and $a\in\mathbb{F}_{q}^{*}$. Then the binomial $x^{n}-a$ is
irreducible in $\mathbb{F}_{q}[x]$ if and only if the following two conditions are satisfied:

($i$) each prime factor of $n$ divides the order $e$ of $a$ in $\mathbb{F}_{q}^{*}$, but not $\frac{q-1}{e}$;

($ii$) $q\equiv1\mod4$ if $n\equiv0 \mod4$.
\end{lemma}
Let $T$ be a finite local commutative ring with maximal ideal $M$. Denote the residue field $T/M$ by $K$ and the
natural projection $T[x]\rightarrow K[x]$ by $\varphi$. Thus the natural ring morphism $T\rightarrow K$ is
simply the restriction of $\varphi$ to the constant polynomials. Let $f$ and $g$ be in $T[x]$. $f$ is a unit
if there exists a polynomial $h$ such that $fh=1$; $f$ is \textit{regular} if $f$ is not a zero divisor; $f$ and $g$ are \textit{coprime}
if $T[x]=(f)+(g)$. By [27], we have the following proposition.

\begin{proposition}
 (Hensel's Lemma) Let $f$ be in $T[x]$ and $\varphi f=\overline{g}_{1}\cdot\cdot\cdot\overline{g}_{n}$,
where $\overline{g}_{1}$,...,$\overline{g}_{n}\in K[x]$ are pairwise coprime. Then there exist $g_{1}$,...,$g_{n}$ in $R[x]$ such that

($i$) $g_{1}$,...,$g_{n}$ are pairwise coprime;

($ii$) $\varphi g_{i}=\overline{g}_{i}, 1\leq i\leq n$;

($iii$) $f=g_{1}\cdot\cdot\cdot g_{n}$.
\end{proposition}

It is obvious that a polynomial in $T[x]$ may not have a unique factorization. However, for regular polynomials over a
finite local commutative ring, the following important factorization property is from [27].

\begin{proposition}
Let $f$ be a regular polynomial in $T[x]$. Then

($i$) $f=\delta g_{1}\cdot\cdot\cdot g_{n}$ where $\delta$ is a unit and $g_{1},...,g_{n}$ are regular primary coprime polynomials;

($ii$) If $f=\delta g_{1}\cdot\cdot\cdot g_{n}=\beta h_{1}\cdot\cdot\cdot h_{m}$ where $\delta$ and $\beta$ are units and ${g_{i}}$
and ${h_{j}}$ are regular primary coprime polynomials, then $n=m$ and, after renumbering, $(h_{i})=(g_{i})$, $1\leq i\leq n$.
\end{proposition}

\section{The case that $x^{n}-\alpha_{0}$ is irreducible in $\mathbb{F}_{p^{m}}[x]$}

In this subsection, we determine the structures of all $\alpha+u\beta$-constacyclic codes of length $np^{s}$ over $R$
when $x^{n}-\alpha_{0}$ is irreducible in $\mathbb{F}_{p^{m}}[x]$. It is clear that any polynomial in $\mathcal{R}$ can
be viewed as a polynomial in $R[x]$.
We start with the following proposition.

\begin{proposition}
Each nonzero polynomial $f(x)$ in $\mathbb{F}_{p^{m}}[x]$ with degree less than $n$ is
invertible in $\mathcal{R}$, that is, there exists $g(x)\in R[x]$ such that $f(x)g(x)\equiv1~(\mod x^{np^{s}}-\alpha-u\beta)$.
\end{proposition}

\begin{proof}
Let $f(x)$ be a nonzero polynomial in $\mathbb{F}_{p^{m}}[x]$ and $0<deg(f)=k<n$. By the division with remainder in
$\mathbb{F}_{p^{m}}[x]$, there exist unique $q(x),r(x)\in\mathbb{F}_{p^{m}}[x]$ such that

\begin{center}
  $x^{n}-\alpha_{0}=f(x)q(x)+r(x)$, $0\leq deg(r)<k$.
\end{center}
Thus
\begin{center}
  $x^{np^{s}}-\alpha=f(x)^{p^{s}}q(x)^{p^{s}}+r(x)^{p^{s}}$.
\end{center}
Noticing that $r(x)\neq0$ since $x^{n}-\alpha_{0}$ is irreducible in $\mathbb{F}_{p^{m}}[x]$.
In the ring $\mathcal{R}$,

\begin{center}
  $f(x)^{p^{s}}q(x)^{p^{s}}+r(x)^{p^{s}}-u\beta=0$.
\end{center}
Let $r(x)^{-1}$ be the inverse element of $r(x)$ if $r(x)$ is invertible in $\mathcal{R}$.
It is easy to check that $u\beta-r(x)^{p^{s}}$ is invertible in $\mathcal{R}$ and
$(u\beta-r(x)^{p^{s}})^{-1}=u\beta r(x)^{-2p^{s}}+r(x)^{-p^{s}}$. It follows that
\begin{center}
$f(x)^{-1}=f(x)^{p^{s}-1}q(x)^{p^{s}}(u\beta r(x)^{-2p^{s}}+r(x)^{-p^{s}})$,
\end{center}
which means that $f(x)$ is invertible in $\mathcal{R}$. Thus we need only prove that all the polynomials
in $\mathbb{F}_{p^{m}}[x]$ of degree less than $k$ are invertible in $\mathcal{R}$.
By induction, it is suffice to consider the case of $k=1$.
In fact, $r(x)\in\mathbb{F}_{p^{m}}^{*}$ is invertible in $\mathcal{R}$ if $k=1$. Thus $f(x)$ is invertible as desired.
\end{proof}

\begin{lemma}
In $\mathcal{R}$, we have $\langle(x^{n}-\alpha_{0})^{p^{s}}\rangle=\langle u\rangle$.
Moreover, $x^{n}-\alpha_{0}$ is nilpotent with nilpotency index $2p^{s}$.
\end{lemma}

\begin{proof}
The desired result follows from the facts that $(x^{n}-\alpha_{0})^{p^{s}}=x^{np^{s}}-\alpha=u\beta$
and $\beta$ is invertible in $\mathcal{R}$. The nilpotency index of $x^{n}-\alpha_{0}$ is obtained from $u^{2}=0$.
\end{proof}

Let $f(x)=f_{1}(x)+uf_{2}(x)$, where $f_{1}(x)$, $f_{2}(x)$ are polynomials over $\mathbb{F}_{p^{m}}$
of degree up to $np^{s}-1$. By Proposition 3.1 and Lemma 3.2, we get the following theorem.

\begin{theorem}
The ring $\mathcal{R}$ is a chain ring whose ideal chain is as follows
\begin{center}
  $\mathcal{R}=\langle1\rangle\supsetneq\langle x^{n}-\alpha_{0}\rangle\supsetneq\cdot\cdot\cdot\supsetneq\langle (x^{n}-\alpha_{0})^{2p^{s}-1}\rangle\supsetneq\langle (x^{n}-\alpha_{0})^{2p^{s}}\rangle=\langle0\rangle$.
\end{center}
 In other words, $(\alpha+u\beta)$-constacyclic codes of length $np^{s}$ over $R$ are precisely the ideals
 $\langle (x^{n}-\alpha_{0})^{i}\rangle$ of $\mathcal{R}$, $0\leq i\leq2p^{s}$. The number of codewords
 of $(\alpha+u\beta)$-constacyclic code $\langle (x^{n}-\alpha_{0})^{i}\rangle$ is $p^{mn(2p^{s}-i)}$.
\end{theorem}

\begin{proof}
As above, let $f(x)=f_{1}(x)+uf_{2}(x)$ be any polynomial in $\mathcal{R}$, where $f_{1}(x),f_{2}(x)\in\mathbb{F}_{p^{m}}[x]$.
In $\mathbb{F}_{p^{m}}[x]$, there exist uniquely $q_{1}(x)$, $r_{1}(x)$ such that
\begin{center}
$f_{1}(x)=q_{1}(x)(x^{n}-\alpha_{0})+r_{1}(x)$,
\end{center}
where $r_{1}(x)=0$ or $0\leq deg(r_{1})<n$. Thus
\begin{eqnarray*}
  f(x) &=& q_{1}(x)(x^{n}-\alpha_{0})+r_{1}(x)+\beta^{-1}(x^{n}-\alpha_{0})^{p^{s}}f_{2}(x) \\
   &=& r_{1}(x)+h(x)(x^{n}-\alpha_{0}),
\end{eqnarray*}
where $h(x)\triangleq q_{1}(x)+\beta^{-1}(x^{n}-\alpha_{0})^{p^{s}-1}f_{2}(x)$.
Since $x^{n}-\alpha_{0}$ is nilpotent from Lemma 3.2, $r_{1}(x)$ is invertible in $\mathcal{R}$ if and only if $f(x)$ is invertible in $\mathcal{R}$.
If $r_{1}(x)$ is non-invertible, we obtain $r_{1}(x)=0$ by Proposition 3.1. It follows that $f(x)\in\langle x^{n}-\alpha_{0}\rangle$.
Thus $f(x)$ is non-invertible if and only if $f(x)\in\langle x^{n}-\alpha_{0}\rangle$. In other words, $\langle x^{n}-\alpha_{0}\rangle$
is the unique maximal ideal of $\mathcal{R}$. By Proposition 2.1, $\mathcal{R}$ is a chain ring with maximal ideal $\langle x^{n}-\alpha_{0}\rangle$.
Therefore, the theorem is proved.
\end{proof}

\begin{remark}
It is obvious that $\alpha_{0}x-1$ is irreducible in $\mathbb{F}_{p^{m}}[x]$,
then Theorem 4.2 in [15] is a corollary of Theorem 3.3. When $\alpha_{0}$
is not a square in $\mathbb{F}_{p^{m}}$, the order of $\alpha_{0}$ is an even number.
By Lemma 2.5 we get $x^{2}-\alpha_{0}$ is irreducible in $\mathbb{F}_{p^{m}}[x]$, then
Theorem 3.3 in [12] is also a corollary of Theorem 3.3.
\end{remark}

We also need to consider the algebraic structure of the dual codes of the $\alpha+u\beta$-constacyclic codes
which are given in Theorem 3.3. By Proposition 2.4, the dual code of an $\alpha+u\beta$-constacyclic code
$\mathcal{C}$ is an $(\alpha+u\beta)^{-1}$-constacyclic code. It is clear that
$(\alpha+u\beta)^{-1}=\alpha^{-1}-u\alpha^{-2}\beta$. Thus
$\mathcal{C^{\bot}}\subseteq R[x]/\langle x^{np^{s}}-\alpha^{-1}-u\alpha^{-2}\beta\rangle$. Let $\alpha_{0}^{'p^{s}}=\alpha^{-1}$.
It is easy to verify that $\alpha_{0}^{'}=\alpha_{0}^{-1}$ and $ord(\alpha_{0})=ord(\alpha_{0}^{-1})$. By Lemma 2.5, $x^{n}-\alpha_{0}$ is irreducible
in $\mathbb{F}_{p^{m}}[x]$ if and only if $x^{n}-\alpha_{0}^{-1}$ is irreducible in $\mathbb{F}_{p^{m}}[x]$. It follows from Theorem 3.3 that
$R[x]/\langle x^{np^{s}}-\alpha^{-1}-u\alpha^{-2}\beta\rangle$ is also a chain ring with the unique maximal ideal
$\langle x^{n}-\alpha_{0}^{-1}\rangle$. Hence,we have the following corollary.

\begin{corollary}
~Let $\mathcal{C}=\langle (x^{n}-\alpha_{0})^{i}\rangle\subseteq\mathcal{R}$ for some $i\in\{0,1,...,2p^{s}\}$
be an $(\alpha+u\beta)$-constacyclic code of length $np^{s}$ over $R$. Then its dual code $\mathcal{C}^{\bot}$ is an
$(\alpha^{-1}-u\alpha^{-2}\beta)$-constacyclic code
\begin{center}
  $\mathcal{C}^{\bot}=\langle (x^{n}-\alpha_{0}^{-1})^{2p^{s}-i}\rangle\subseteq R[x]/\langle x^{np^{s}}-\alpha^{-1}-u\alpha^{-2}\beta\rangle$.
\end{center}
\end{corollary}

\begin{proof}
Let $\mathcal{C}^{\bot}=\langle (x^{n}-\alpha_{0}^{-1})^{j}\rangle\subseteq R[x]/\langle x^{np^{s}}-\alpha^{-1}-u\alpha^{-2}\beta\rangle$
be the dual code of the constacyclic code $\mathcal{C}=\langle (x^{n}-\alpha_{0})^{i}\rangle\subseteq R[x]/\langle x^{np^{s}}-\alpha-u\beta\rangle$,
where $0\leq i,j\leq2p^{s}$. By Proposition 2.3, $|\mathcal{C}||\mathcal{C}^{\bot}|=p^{mn(4p^{s}-i-j)}=|R|^{np^{s}}=p^{2mnp^{s}}$.
Therefore, $j=2p^{s}-i$.
\end{proof}

As a corollary of Theorem 3.3 and corollary 3.5, we can obtain the following result.

\begin{corollary}
The ideal $\langle u\rangle=\langle (x^{n}-\alpha_{0})^{p^{s}}\rangle$ is the
unique self-dual $(\alpha+u\beta)$-constacyclic code of length $np^{s}$ over $R$.
\end{corollary}

\begin{proof}
Let $ua(x)$, $ub(x)$ be two arbitrary elements of the ideal $\langle u\rangle$.
Denote $a(x)$ and $b(x)$ are corresponding to the codewords $\mathbf{a}=(a_{0},a_{1},...,a_{np^{s}-1})\in R^{np^{s}}$,
$\mathbf{b}=(b_{0},b_{1},...,b_{np^{s}-1})\in R^{np^{s}}$ respectively. Note that $u^{2}=0$. Then
\begin{center}
  $\langle u\mathbf{a},u\mathbf{b}\rangle=u^{2}\sum\limits_{i=0}^{n-1}a_{i}b_{i}=0$.
\end{center}
This implies that $\langle u\rangle\subseteq\langle u\rangle^{\bot}$. By Proposition 2.3, we have
\begin{center}
  $|\langle u\rangle^{\bot}|=\frac{|R|^{np^{s}}}{|\langle u\rangle|}=\frac{p^{2mnp^{s}}}{p^{mnp^{s}}}=p^{mnp^{s}}=|\langle u\rangle|$.
\end{center}
This means $\langle u\rangle=\langle u\rangle^{\bot}$.
Hence $\langle u\rangle$ is a self-dual $(\alpha+u\beta)$-constacyclic code.
Now suppose that $\mathcal{C}=\langle (x^{n}-\alpha_{0})^{i}\rangle\subseteq\mathcal{R}$ is a self-dual $(\alpha+u\beta)$-constacyclic code of length $np^{s}$ over $R$.It follows from Corollary 3.5 that
\begin{center}
  $|\langle (x^{n}-\alpha_{0})^{i}\rangle|=|\mathcal{C}|=|\mathcal{C}^{\bot}|=|\langle (x^{n}-\alpha_{0}^{-1})^{2p^{s}-i}\rangle|$.
\end{center}
Thus $p^{mn(2p^{s}-i}=p^{mni}$, i.e., $i=p^{s}$.
We have $\mathcal{C}=\langle (x^{n}-\alpha_{0})^{p^{s}}\rangle=\langle u\rangle$ by Lemma 3.2.
Then the uniqueness is proved.
\end{proof}

\section{The case that $x^{n}-\alpha_{0}$ is reducible in $\mathbb{F}_{p^{m}}[x]$}

We now consider the case that $x^{n}-\alpha_{0}$ is reducible in $\mathbb{F}_{p^{m}}[x]$. Since $\gcd(n,p)=1$, $x^{n}-\alpha_{0}$ has no repeated
fators. Let $f_{1}(x)$,$f_{2}(x)$,...,$f_{r}(x)$ be pairwise coprime monic irreducible polynomials
in $\mathbb{F}_{p^{m}}[x]$ such that $x^{n}-\alpha_{0}=f_{1}(x)f_{2}(x)\cdot\cdot\cdot f_{r}(x)$. Then,
\begin{center}
  $x^{np^{s}}-\alpha-u\beta=f_{1}(x)^{p^{s}}f_{2}(x)^{p^{s}}\cdot\cdot\cdot f_{r}(x)^{p^{s}}-u\beta$.
\end{center}
Since $f_{1}(x)^{p^{s}}$ and $f_{2}(x)^{p^{s}}\cdot\cdot\cdot f_{r}(x)^{p^{s}}$ are coprime in $\mathbb{F}_{p^{m}}[x]$,
we have there exist $\nu_{1}(x),\omega_{1}(x)\in\mathbb{F}_{p^{m}}[x]$ such that
\begin{center}
  $\nu_{1}(x)f_{1}(x)^{p^{s}}+\omega_{1}(x)f_{2}(x)^{p^{s}}\cdot\cdot\cdot f_{r}(x)^{p^{s}}=1$,
\end{center}
 where $deg(\omega_{1})<deg(f_{1}^{p^{s}})$. Thus, we obtain
\begin{center}
  $x^{np^{s}}-\alpha-u\beta=(f_{1}(x)^{p^{s}}-u\beta\omega_{1}(x))(f_{2}(x)^{p^{s}}\cdot\cdot\cdot f_{r}(x)^{p^{s}}-u\beta\nu_{1}(x))$.
\end{center}
It is a routine to verify that $f_{1}(x)^{p^{s}}-u\beta\omega_{1}(x)$ and $f_{2}(x)^{p^{s}}\cdot\cdot\cdot f_{r}(x)^{p^{s}}-u\beta\nu_{1}(x)$
are coprime in $R[x]$. Next we consider the factorization of
$f_{2}(x)^{p^{s}}\cdots f_{r}(x)^{p^{s}}-u\beta\nu_{1}(x)$. Since
$f_{2}(x)^{p^{s}}$ and $f_{3}(x)^{p^{s}}\cdot\cdot\cdot f_{r}(x)^{p^{s}}$ are coprime in $F_{p^{m}}[x]$, there exist
$\nu_{2}(x),\omega_{2}(x)\in\mathbb{F}_{p^{m}}[x]$ such that $deg(\omega_{2})<deg(f_{2}^{p^{s}})$ and
\begin{center}
  $\nu_{2}(x)f_{2}(x)^{p^{s}}+\omega_{2}(x)f_{3}(x)^{p^{s}}\cdot\cdot\cdot f_{r}(x)^{p^{s}}=1$.
\end{center}
Thus,

\begin{align*}
   &  f_{2}(x)^{p^{s}}\cdots f_{r}(x)^{p^{s}}-u\beta\nu_{1}(x)\\
  = & (f_{2}(x)^{p^{s}}-u\beta\nu_{1}(x)\omega_{2}(x))(f_{3}(x)^{p^{s}}\cdots f_{r}(x)^{p^{s}}-u\beta\nu_{1}(x)\nu_{2}(x)).
\end{align*}
\noindent Repeating this process, we get
\begin{eqnarray*}
  x^{np^{s}}-\alpha-u\beta &=& (f_{1}(x)^{p^{s}}-u\beta\omega_{1}(x))(f_{2}(x)^{p^{s}}-u\beta\nu_{1}(x)\omega_{2}(x)) \\
   && \cdots(f_{r}(x)^{p^{s}}-u\beta\nu_{1}(x)\nu_{2}(x)\cdots\nu_{r-1}(x)).
\end{eqnarray*}
\noindent Let $h_{j}(x)=f_{j}(x)^{p^{s}}+ug_{j}(x)$ ($1\leq j\leq r$) where $g_{1}(x)=-\beta\omega_{1}(x)$,
$g_{j}(x)=-\beta\nu_{1}(x)\cdot\cdot\cdot\nu_{j-1}(x)\omega_{j}(x)$, for $2\leq j\leq r-1$,
$g_{r}(x)=-\beta\nu_{1}(x)\cdot\cdot\cdot\nu_{r-1}(x)$. Then $x^{np^{s}}-\alpha-u\beta=h_{1}(x)h_{2}(x)\cdot\cdot\cdot h_{r}(x)$ and $h_{1}(x)$,$h_{2}(x)$,...,$h_{r}(x)$
are pairwise coprime in $R[x]$.

Using the notations above, we obtain the following result.
\begin{lemma}
For any $1\leq i\leq r$, $f_{i}(x)$ and $g_{i}(x)$ are coprime in $\mathbb{F}_{p^{m}}[x]$.
\end{lemma}
\begin{proof}
As above, $x^{np^{s}}-\alpha-u\beta=\prod\limits_{i=1}^{r}(f_{i}(x)^{p^{s}}+ug_{i}(x))$.
Expanding the right side of the equation and comparing two sides of the equation, we can obtain that
$v_{i}(x)g_{i}(x)+w_{i}(x)f_{i}(x)=-\beta$ for some $v_{i}(x), w_{i}(x)\in\mathbb{F}_{p^{m}}[x]$ and for all $1\leq i\leq r$.
The conclusion is obtained from $\beta\neq 0$.
\end{proof}

For any integer $j$, $1\leq j\leq r$, we assume that $\deg(f_{j}(x))=d_{j}$ and denote
$H_{j}(x)=h_{1}(x)\cdots h_{j-1}(x)h_{j+1}(x)\cdots h_{r}(x)$. It is obvious that $H_{j}(x)$ and
$h_{j}(x)$ are copime in $R[x]$. Hence there exist $s_{j}(x),t_{j}(x)\in R[x]$ such that
\begin{center}
  $s_{j}(x)H_{j}(x)+t_{j}(x)h_{j}(x)=1$.
\end{center}
Let $\varepsilon_{j}(x)=s_{j}(x)H_{j}(x)\mod(x^{np^{s}}-\alpha-u\beta)$.
Then by the Chinese remainder theorem for commutative rings with identity, we deduce the following conclusion.
\begin{lemma}
In the ring $\mathcal{R}$, the following statements hold for all $1\leq j\neq l\leq r$:

($i$)~$\varepsilon_{1}(x)+\cdot\cdot\cdot+\varepsilon_{r}(x)=1$;

($ii$)~$\varepsilon_{j}(x)^{2}=\varepsilon_{j}(x)$;

($iii$)~$\varepsilon_{j}(x)\varepsilon_{l}(x)=0$.
\end{lemma}
\begin{proof}
($i$)~By the definition of $\varepsilon_{j}(x)$, we have
\begin{center}
  $\varepsilon_{1}+\varepsilon_{2}+\cdot\cdot\cdot+\varepsilon_{r}\equiv \sum\limits_{i\neq j}s_{i}H_{i}+1-t_{j}h_{j}~(\!\!\!\!\mod x^{np^{s}}-\alpha-u\beta)$.
\end{center}
Thus $\varepsilon_{1}+\varepsilon_{2}+\cdot\cdot\cdot+\varepsilon_{r}\equiv 1~(\!\!\!\!\mod h_{j}), 1\leq j\leq r$.
By Proposition 2.7,
\begin{center}
$\varepsilon_{1}+\varepsilon_{2}+\cdot\cdot\cdot+\varepsilon_{r}\equiv 1~(\!\!\!\!\mod x^{np^{s}}-\alpha-u\beta)$.
\end{center}
 Therefore, the first statement is valid.

($ii$)~Since
\begin{center}
$0\equiv s_{j}H_{j}h_{j}\equiv h_{j}-t_{j}h_{j}^{2}~(\!\!\!\!\mod x^{np^{s}}-\alpha-u\beta)$,
\end{center}
 we have
 \begin{center}
 $h_{j}\equiv t_{j}h_{j}^{2}~(\!\!\!\!\mod x^{np^{s}}-\alpha-u\beta)$.
 \end{center}
Thus
\begin{center}
$\varepsilon_{j}^{2}\equiv1-2t_{j}h_{j}+t_{j}^{2}h_{j}^{2}\equiv 1-t_{j}h_{j}\equiv \varepsilon_{j}~(\!\!\!\!\mod x^{np^{s}}-\alpha-u\beta)$.
\end{center}

($iii$)~It is obvious.
\end{proof}

By Lemma 4.2, we get another expression of $\mathcal{R}$:
\begin{lemma}
$\mathcal{R}=\mathcal{R}_{1}\oplus\cdot\cdot\cdot\oplus\mathcal{R}_{r}$, where $\mathcal{R}_{j}=\mathcal{R}\varepsilon_{j}(x)$
with $\varepsilon_{j}(x)$ as its multiplicative identity.
\end{lemma}
\begin{proof}
Since $\mathcal{R}_{j}=\mathcal{R}\varepsilon_{j}(x)$ is a subring of $\mathcal{R}$,
$\mathcal{R}_{1}+\cdots+\mathcal{R}_{r}\subseteq\mathcal{R}$. For any $f(x)\in\mathcal{R}$,
$f(x)=f(x)\varepsilon_{1}(x)+\cdots+f(x)\varepsilon_{r}(x)$. Therefore
$f(x)\in\mathcal{R}_{1}+\cdots+\mathcal{R}_{r}$, that is, $\mathcal{R}\subseteq\mathcal{R}_{1}+\cdots+\mathcal{R}_{r}$.
Let $0=a_{1}(x)\varepsilon_{1}(x)+a_{2}(x)\varepsilon_{2}(x)+\cdot\cdot\cdot+a_{r}(x)\varepsilon_{r}(x)$,
where $a_{1}(x),a_{2}(x),\ldots,a_{r}(x)\in\mathcal{R}$. It follows from
Lemma 4.2 ($iii$) that $0=a_{j}(x)\varepsilon_{j}(x)$ for all $1\leq j\leq r$. Thus $\mathcal{R}=\mathcal{R}_{1}\oplus\cdots\oplus\mathcal{R}_{r}$.
\end{proof}

Denote $\mathcal{K}_{j}=R[x]/\langle h_{j}(x)\rangle$. We have the following isomorphism.
\begin{lemma}
For any integer $j\in\{1,2,\ldots,r\}$, $\mathcal{K}_{j}$ and $\mathcal{R}_{j}$ are isomorphic as rings.
\end{lemma}
\begin{proof}
We define a mapping $\phi_{j}: \mathcal{K}_{j}\rightarrow\mathcal{R}_{j}$ as follows:
\begin{center}
$a(x)\mapsto\varepsilon_{j}(x)a(x) \mod(x^{np^{s}}-\alpha-u\beta)$.
\end{center}
For $a(x), b(x)\in\mathcal{K}_{j}$. If $a(x)=b(x)$, then there exists $q(x)\in R[x]$
such that
\begin{center}
$a(x)-b(x)=q(x)h_{j}(x)$.
\end{center}
 Thus,
 \begin{center}
 $a(x)\varepsilon_{j}(x)-b(x)\varepsilon_{j}(x)=q(x)h_{j}(x)\varepsilon_{j}(x)$.
 \end{center}
Since $h_{j}(x)\varepsilon_{j}(x)\equiv s_{j}H_{j}h_{j}\equiv0 \mod(x^{np^{s}}-\alpha-u\beta)$, we get
$\phi_{j}(a)=\phi_{j}(b)$, which means $\phi_{j}$ is well-defined.
If $\phi_{j}(a)=\phi_{j}(b)$, there exist $q^{'}(x)\in R[x]$ such that
$a(x)\varepsilon_{j}(x)-b(x)\varepsilon_{j}(x)=q^{'}(x)(x^{np^{s}}-\alpha-u\beta)$.
It follows from $\varepsilon_{j}(x)\equiv 1-t_{j}(x)h_{j}(x) \mod(x^{np^{s}}-\alpha-u\beta)$ that
$a(x)\equiv b(x)\mod h_{j}(x)$. Then $\phi_{j}$ is injection.
It is obvious that $\phi_{j}$
is a surjective ring homomorphism. Therefore, $\phi_{j}$ is a ring isomorphism.
\end{proof}

We can construct a mapping $\phi$ from $\mathcal{K}_{1}\times\cdots\times\mathcal{K}_{r}$ onto $\mathcal{R}$
via $\phi_{j}$ as follows:
\begin{center}
  $\phi(a_{1}(x),\ldots,a_{r}(x))=\sum\limits_{j=1}^{r}\phi_{j}(a_{j}(x))=\sum\limits_{j=1}^{r}\varepsilon_{j}(x)a_{j}(x)\mod(x^{np^{s}}-\alpha-u\beta)$.
\end{center}
It is easy to verify that $\phi$ is a ring isomorphism. Therefore, we have the following result.
\begin{lemma}
$\mathcal{K}_{1}\times\cdot\cdot\cdot\times\mathcal{K}_{r}$ is isomorphic to $\mathcal{R}$.
 \end{lemma}

\begin{theorem}
Let $\mathcal{C}$ be a subset of $\mathcal{R}$. Then $\mathcal{C}$ is an $\alpha+u\beta$-contacyclic code of length $np^{s}$
over $\mathrm{R}$ if and only if for each integer $j$, $1\leq j\leq r$, there is a unique ideal $\mathcal{C}_{j}$ of
$\mathcal{K}_{j}$ such that $\mathcal{C}=\bigoplus\limits_{j=1}^{r}\varepsilon_{j}(x)\mathcal{C}_{j}~(\!\!\!\!\mod x^{np^{s}}-\alpha-u\beta)$.
\end{theorem}
\begin{proof}
By Lemma 4.5, we know that $\mathcal{C}$ is an ideal of $\mathcal{R}$ if and only if there is a unique ideal $\mathcal{I}$ of the
ring $\mathcal{K}_{1}\times\cdot\cdot\cdot\times\mathcal{K}_{r}$ such that $\phi(\mathcal{I})=\mathcal{C}$.
Furthermore, by classical ring theory we see that $\mathcal{I}$ is an ideal of $\mathcal{K}_{1}\times\cdot\cdot\cdot\times\mathcal{K}_{r}$
if and only if for each integer $j$, $1\leq j\leq r$, there is a unique ideal $\mathcal{C}_{j}$ of $\mathcal{K}_{j}$ such that
\begin{center}
$\mathcal{I}=\mathcal{C}_{1}\times\cdot\cdot\cdot\times\mathcal{C}_{r}=\{(a_{1},...,a_{r})|a_{j}\in\mathcal{C}_{j},j=1,...,r\}$.
\end{center}
When this condition is satisfied, we have
\begin{center}
  $\mathcal{C}=\phi(\mathcal{I})=\{\sum\limits_{j=1}^{r}\varepsilon_{j}(x)a_{j}|a_{j}\in\mathcal{C}_{j}, j=1,...,r\}$.
\end{center}
Hence $\mathcal{C}=\bigoplus\limits_{j=1}^{r}\varepsilon_{j}(x)\mathcal{C}_{j}$.
\end{proof}

In order to determine all $\alpha+u\beta$-constacyclic codes over $R$, by Theorem 4.6, we need only to study the ideals of $\mathcal{K}_{j}=R[x]/\langle f_{j}(x)^{p^{s}}+ug_{j}(x)\rangle$.

\begin{proposition}
All the nonzero polynomials in $\mathbb{F}_{p^{m}}[x]$ of degree less than $d_{j}$ are invertible in $\mathcal{K}_{j}$.
\end{proposition}
\begin{proof}
The proof is similar to that of Proposition 3.1.
\end{proof}

Furthermore, we have the following result.
\begin{proposition}
In the ring $\mathcal{K}_{j}$, $\langle f_{j}(x)^{p^{s}}\rangle=\langle u\rangle$ and $f_{j}(x)$ is a nilpotent with nipotency index $2p^{s}$.
\end{proposition}
\begin{proof}
Since $f_{j}(x)^{p^{s}}+ug_{j}(x)=0$ in $\mathcal{K}_{j}$, $\langle f_{j}(x)^{p^{s}}\rangle\subseteq\langle u\rangle$.
It follows from $u^{2}=0$ that $f_{j}(x)$ is a nilpotent with nipotency index $2p^{s}$ in $\mathcal{K}_{j}$.
By division with remainder in $\mathbb{F}_{p^{m}}[x]$, there exist $q(x)$, $r(x)\in\mathbb{F}_{p^{m}}[x]$ such that
\begin{center}
  $g_{j}(x)=f_{j}(x)q(x)+r(x)$,
\end{center}
 where $0\leq deg(r)<d_{j}$ or $r(x)=0$.
Since $\gcd(f_{j}(x),g_{j}(x))=1$, we obtain $r(x)\neq0$. It follows from Proposition 4.7 that $r(x)$ is
invertible in $\mathcal{K}_{j}$. From the nilpotency of $f_{j}(x)$, we obtain $g_{j}(x)$ is invertible in $\mathcal{K}_{j}$ and
\begin{center}
  $g_{j}(x)^{-1}=r(x)^{-1}-r(x)^{-2}f_{j}(x)q(x)-r(x)^{-3}f_{j}(x)^{2}q(x)^{2}$

  $-\cdot\cdot\cdot-r(x)^{-2p^{s}}f_{j}(x)^{2p^{s}-1}q(x)^{2p^{s}-1}$.
\end{center}
Thus, $u=-f_{j}(x)^{p^{s}}g_{j}(x)^{-1}$, i.e., $\langle u\rangle\subseteq\langle f_{j}(x)^{p^{s}}\rangle$.
\end{proof}

\begin{theorem}
The ring $\mathcal{K}_{j}$ is a chain ring whose ideal chain is as follows
\begin{center}
  $\mathcal{K}_{j}=\langle1\rangle\supsetneq\langle f_{j}(x)\rangle\supsetneq\cdot\cdot\cdot\supsetneq\langle f_{j}(x)^{2p^{s}-1}\rangle\supsetneq\langle f_{j}(x)^{2p^{s}}\rangle=\langle0\rangle$.
\end{center}
Each ideal $\langle f_{j}(x)^{i_{j}}\rangle$ has $p^{d_{j}m(2p^{s}-i_{j})}$ elements, where $0\leq i_{j}\leq2p^{s}$.
\end{theorem}
\begin{proof}
Let $k(x)\in\mathcal{K}_{j}$, we can write uniquely $k(x)=k_{1}(x)+uk_{2}(x)$ with $k_{1}(x), k_{2}(x)\in\mathbb{F}_{p^{m}}[x]$.
Using division with remainder in $\mathbb{F}_{p^{m}}[x]$, there exist $q(x), r_{1}(x)\in\mathbb{F}_{p^{m}}[x]$ such that
\begin{center}
  $k_{1}(x)=q(x)f_{j}(x)+r_{1}(x)$,
\end{center}
where $r_{1}(x)=0$ or $0\leq deg(r_{1})<d_{j}$.
Thus
\begin{center}
$k(x)=q(x)f_{j}(x)+r_{1}(x)+uk_{2}(x)=(q(x)-f_{j}(x)^{p^{s}-1}g_{j}(x)^{-1}k_{2}(x))f_{j}(x)+r_{1}(x)$
\end{center}
If $r_{1}(x)$ is invertible in $\mathcal{K}_{j}$, then $k(x)$ is invertible in $\mathcal{K}_{j}$ as $f_{j}(x)$ is a nilpotent. Otherwise, it follows from Proposition 4.7 that $r_{1}(x)=0$. Thus,
$k(x)\in\langle f_{j}(x)\rangle$ and so $g(x)$ is non-invertible in $\mathcal{K}_{j}$. Furthermore, $k(x)$ is non-invertible in $\mathcal{K}_{j}$
if and only if $k(x)\in\langle f_{j}(x)\rangle$, which means $\langle f_{j}(x)\rangle$
is the unique maximal ideal of $\mathcal{K}_{j}$. By Proposition 2.1, $\mathcal{K}_{j}$ is a chain ring.
\end{proof}

\begin{corollary}
Every $\alpha+u\beta$-constacyclic code $\mathcal{C}$ of length $np^{s}$ over $\mathrm{R}$ is
\begin{center}
  $\mathcal{C}=\bigoplus\limits_{j=1}^{r}\varepsilon_{j}(x)\langle f_{j}(x)^{i_{j}}\rangle~\!\!\!\!\mod(x^{np^{s}}-\alpha-u\beta)$,
\end{center}
 where $0\leq i_{j}\leq2p^{s}$.
The number of codewords of $\mathcal{C}$ is equal to $p^{\sum\limits_{j=1}^{r}d_{j}m(2p^{s}-i_{j})}$.
 Furthermore, the number of $\alpha+u\beta$-constacyclic code over $\mathrm{R}$ of length $np^{s}$ is equal to $(2p^{s}+1)^{r}$.
\end{corollary}

For any polynomial $h(x)=\sum\limits_{i=0}^{d}c_{i}x^{i}\in\mathrm{R}[x]$ of degree $d\geq1$. Recall that the \emph{reciprocal polynomial} of $h(x)$ is defined as $\widetilde{h}(x)=\widetilde{h(x)}=x^{d}h(\frac{1}{x})=\sum\limits_{i=0}^{d}c_{i}x^{d-i}$
and $h(x)$ is said to be \emph{self-reciprocal} if $\widetilde{h}(x)=\delta h(x)$ for some unit $\delta$ in $R$. It is known that $\widetilde{\widetilde{h}(x)}=h(x)$ if $h(0)\neq0$, and
$\widetilde{h_{1}(x)h_{2}(x)}=\widetilde{h_{1}}(x)\widetilde{h_{2}}(x)$ if $h_{1}(x)$, $h_{2}(x)$ are not zero divisor. Using the notations above, we have
\begin{center}
$x^{np^{s}}-(\alpha+u\beta)^{-1}=-(\alpha+u\beta)^{-1}\widetilde{h_{1}}(x)\widetilde{h_{2}}(x)\cdot\cdot\cdot\widetilde{h_{r}}(x)$.
\end{center}
Since $-(\alpha+u\beta)^{-1}$ is a unit in $\mathrm{R}$,
\begin{center}
$R[x]/\langle x^{np^{s}}-(\alpha+u\beta)^{-1}\rangle=\mathrm{R}[x]/\langle\widetilde{h_{1}}(x)\widetilde{h_{2}}(x)\cdot\cdot\cdot\widetilde{h_{r}}(x)\rangle$.
\end{center}
It is a fact that $h_{1}(x)$, $h_{2}(x)$,...,$h_{r}(x)$ are pairwise coprime if and only if
$\widetilde{h_{1}}(x)$, $\widetilde{h_{2}}(x)$,...,$\widetilde{h_{r}}(x)$ are pairwise coprime.
Using the Chinese remainder theorem, we get
\begin{center}
  $R[x]/\langle x^{np^{s}}-(\alpha+u\beta)^{-1}\rangle\cong R[x]/\langle\widetilde{h_{1}}(x)\rangle\oplus \cdot\cdot\cdot\oplus R[x]/\langle\widetilde{h_{r}}(x)\rangle$
\end{center}
Next we discuss this isomorphism in detail. Some notations are given here:

\begin{eqnarray*}
\widehat{\mathcal{R}} &=& R[x]/\langle x^{np^{s}}-(\alpha+u\beta)^{-1}\rangle; \\
\widehat{\mathcal{K}}_{j} &=& R[x]/\widetilde{h_{j}}(x), 1\leq j\leq r.
\end{eqnarray*}
We define a map $\tau: \mathcal{R}\rightarrow\widehat{\mathcal{R}}$ as follows:
\begin{center}
 $\tau(a(x))=a(x^{-1})$, $(\forall a(x)\in\mathcal{R})$.
\end{center}
Here, $x^{-1}=(\alpha+u\beta)x^{np^{s}-1}$ in $\widehat{\mathcal{R}}$. Then one can easily verify that $\tau$ is a ring isomorphism from
$\mathcal{R}$ onto $\widehat{\mathcal{R}}$.

As what we have discussed on $\mathcal{R}$, we define
\begin{center}
$\widehat{\varepsilon}_{j}(x)\equiv v_{j}(x^{-1})H_{j}(x^{-1})\equiv 1-w_{j}(x^{-1})h_{j}(x^{-1})~(\!\!\!\!\mod x^{np^{s}}-(\alpha+u\beta)^{-1})$.
\end{center}
 Then we have some lemmas in the ring $\widehat{\mathcal{R}}$ which is similar to $ \mathcal{R}$.
\begin{lemma}
(1) $\widehat{\varepsilon}_{1}(x)+\cdot\cdot\cdot+\widehat{\varepsilon}_{r}(x)=1$ ,
$\widehat{\varepsilon}_{j}(x)^{2}=\widehat{\varepsilon}_{j}(x)$ and $\widehat{\varepsilon}_{j}(x)\widehat{\varepsilon}_{l}(x)=0$
in the ring $\widehat{\mathcal{R}}$ for all $1\leq j\neq l\leq r$.

(2) $\widehat{\mathcal{R}}=\widehat{\mathcal{R}}_{1}\oplus\cdot\cdot\cdot\oplus\widehat{\mathcal{R}}_{r}$
where $\widehat{\mathcal{R}}_{j}=\widehat{\mathcal{R}}\widehat{\varepsilon}_{j}(x)$ with $\widehat{\varepsilon}_{j}(x)$
as its multiplicative identity and satisfies $\widehat{\mathcal{R}}_{j}\widehat{\mathcal{R}}_{l}={0}$ for all $1\leq j\neq l\leq r$.

(3) For any integer $j$,$1\leq j\leq r$, for any $a(x)\in\widehat{\mathcal{K}}_{j}$ we define
\begin{center}
  $\psi_{j}: a(x)\mapsto\widehat{\varepsilon}_{j}(x)a(x)~\mod(x^{np^{s}}-(\alpha+u\beta)^{-1})$.
\end{center}
Then $\psi_{j}$ is a ring isomorphism from $\widehat{\mathcal{K}}_{j}$ onto $\widehat{\mathcal{R}}_{j}$.

(4) For any $a_{j}(x)\in\widehat{\mathcal{K}}_{j}$ for $j=1,...,r$, define
\begin{center}
  $\psi(a_{1}(x),...,a_{r}(x))=\sum\limits_{j=1}^{r}\psi_{j}(a_{j}(x))=\sum\limits_{j=1}^{r}\widehat{\varepsilon}_{j}(x)a(x)~(\!\!\!\!\mod x^{np^{s}}-(\alpha+u\beta)^{-1})$.
\end{center}
Then $\psi$ is a ring isomorphism from $\widehat{\mathcal{K}}_{1}\times\cdot\cdot\cdot\times\widehat{\mathcal{K}}_{r}$ onto $\widehat{\mathcal{R}}$.
\end{lemma}

We now define a map $\tau_{j}:\mathcal{K}_{j}\rightarrow\widehat{\mathcal{K}}_{j}$ by $\tau_{j}(a(x))=a(x^{-1})$ for any
$a(x)\in\mathcal{K}_{j}$. It is easy to verify that $\tau_{j}$ is a ring isomorphism from
$\mathcal{K}_{j}$ onto $\widehat{\mathcal{K}}_{j}$. By Theorem 4.7 we know the ideals in the ring
$\mathcal{K}_{j}$ are of the forms $\langle f_{j}(x)^{i_{j}}\rangle$, $0\leq i_{j}\leq2p^{s}$.
Hence we obtain that every ideal in the ring $\widehat{\mathcal{K}}_{j}$ is of the form
$\langle f_{j}(x^{-1})^{i_{j}}\rangle$. Thus the ideal of $\widehat{\mathcal{R}}$ can be given by
$\mathcal{C}=\sum\limits_{j=1}^{r}\widehat{\varepsilon}_{j}(x)\langle f_{j}(x^{-1})^{i_{j}}\rangle$.

\begin{lemma}
Let $\textbf{a}=(a_{0},a_{1},...,a_{np^{s}-1})$ and $\textbf{b}=(b_{0},b_{1},...,b_{np^{s}-1})$, where $a_{i}, b_{i}\in\mathrm{R}$ for all $i=0,1,...,np^{s}-1$. Let $a(x)=\sum\limits_{i=0}^{np^{s}-1}a_{i}x^{i}\in\mathcal{R}$, $b(x)=\sum\limits_{i=0}^{np^{s}-1}b_{i}x^{i}\in\widehat{\mathcal{R}}$.
 If $\tau(a(x))b(x)=0$ in $\widehat{\mathcal{R}}$, then $\textbf{a}\cdot\textbf{b}=0$.
\end{lemma}
\begin{proof}
Since $x^{-1}=(\alpha+u\beta)x^{np^{s}-1}$ in $\widehat{\mathcal{R}}$,
we have
$\tau(a(x))=a_{0}+a_{1}(\alpha+u\beta)x^{np^{s}-1}+a_{2}(\alpha+u\beta)x^{np^{s}-2}+\cdot\cdot\cdot+a_{np^{s}-1}(\alpha+u\beta)x$.
Furthermore, $\tau(a(x))b(x)\!\!=\sum\limits_{i=0}^{np^{s}-1}a_{i}b_{i}+c_{1}x+\cdot\cdot\cdot+c_{np^{s}-1}x^{np^{s}-1}$.
Thus $\textbf{a}\cdot \textbf{b}=0$ from $\tau(a(x))b(x)=0$.
\end{proof}

\begin{theorem}
Let $\mathcal{C}$ be an $(\alpha+u\beta)$-constacyclic code over $\mathrm{R}$ of length $np^{s}$ with $\mathcal{C}=\sum\limits_{j=1}^{r}\varepsilon_{j}(x)\langle f_{j}(x)^{i_{j}}\rangle~(\!\!\!\!\mod x^{np^{s}}-\alpha-u\beta)$.
Then the dual code $\mathcal{C}^{\bot}$ of $\mathcal{C}$ is an $(\alpha+u\beta)^{-1}$-constacyclic code over $\mathrm{R}$ of length $np^{s}$ with
\begin{center}
  $\mathcal{C}^{\bot}=\sum\limits_{j=1}^{r}\widehat{\varepsilon}_{j}(x)\langle f_{j}(x^{-1})^{2p^{s}-i_{j}}\rangle~\!\!\!\!\mod(x^{np^{s}}-(\alpha+u\beta)^{-1})$.
\end{center}
\end{theorem}
\begin{proof}
It is easy to prove that $\tau(\varepsilon_{j}(x))\widehat{\varepsilon}_{l}(x)=0$ in the residue ring
$\widehat{\mathcal{R}}$  if $j\neq l$. From the ring isomorphism
$\tau_{j}: \mathcal{K}_{j}\rightarrow\widehat{\mathcal{K}}_{j}$, we know that $0=\tau_{j}(f_{j}(x)^{2p^{s}})=f_{j}(x^{-1})^{2p^{s}}$.
Let $\mathcal{D}=\sum\limits_{j=1}^{r}\widehat{\varepsilon}_{j}(x)\langle f_{j}(x^{-1})^{2p^{s}-i_{j}}\rangle$. Then
\begin{center}
$\tau(\mathcal{C})\cdot \mathcal{D}=(\sum\limits_{l=1}^{r}\varepsilon_{l}(x^{-1})\langle f_{l}(x^{-1})^{s_{l}}\rangle)(\sum\limits_{j=1}^{r}\widehat{\varepsilon}_{l}(x)\langle f_{j}(x^{-1})^{2p^{s}-i_{j}}\rangle)=0$.
\end{center}
Thus, $\mathcal{D}\subseteq\mathcal{C}^{\bot}$. Moreover from $|\mathcal{D}|=p^{\sum\limits_{j=1}^{r}d_{j}mi_{j}}=|\mathcal{C}^{\bot}|$, we get $\mathcal{D}=\mathcal{C}^{\bot}$.
\end{proof}

\end{document}